\documentclass[preprint,12pt]{elsarticle}

\usepackage{amssymb}
\usepackage{amsthm}

\journal{Elsevier journal}

\usepackage[fleqn]{amsmath}
\usepackage{amsfonts}
\usepackage[mathscr]{eucal}
\def\P{\mathbb{P}}
\def\R{\mathbb{R}}
\def\E{\mathbb{E}}
\def\y{{\boldsymbol{y}}}
\def\x{{\boldsymbol{x}}}
\def\b{{\boldsymbol{b}}}
\def\h{{\boldsymbol{h}}}
\def\e{{\boldsymbol{e}}}
\def\d{{\boldsymbol{d}}}

\def\c{{\boldsymbol{c}}}

\def\w{{\boldsymbol{w}}}
\def\owv{\overline{\boldsymbol{w}}}
\def\0{{\bf{0}}}
\def\1{{\bf{1}}}
\def\G{{\bf{G}}}

\def\I{{\bf{I}}}
\def\A{{\bf{A}}}

\def\eAv{\widehat{\bf{A}}}
\def\vbeta{{\boldsymbol{\beta}}}
\def\valpha{{\boldsymbol{\alpha}}}

\def\evmu{\widehat{\boldsymbol{\mu}}}
\def\ebv{\widehat{\boldsymbol{b}}}
\def\ecv{\widehat{\boldsymbol{c}}}
\def\tcv{\widetilde{\boldsymbol{c}}}
\def\tbv{\widetilde{\boldsymbol{b}}}
\def\evalpha{{\widehat{\boldsymbol{\alpha}}}}
\def\eR{\widehat{R}}
\def\eb{\widehat{b}}
\def\ec{\widehat{c}}

\def\tc{\widetilde{c}}
\def\oc{\overline{c}}
\def\tb{\widetilde{b}}

\def\sign{{\rm sign}}

\def\etheta{\widehat{\theta}}
\def\ttheta{\widetilde{\theta}}

\def\obeta{\overline{\beta}}

\def\oK{\overline{K}}

\def\eK{\widehat{K}}
\def\oE{\overline{E}}
\def\oF{\overline{F}}
\def\oG{\overline{G}}
\def\eR{\widehat{R}}
\def\esigma{\widehat{\sigma}}
\def\median{{\rm median}}

\def\det{{\rm det}}
\def\ealpha{\widehat{\alpha}}

\def\RAS{R_{\rm AS}}

\newtheorem{theorem}{Theorem}
\newtheorem{lemma}{Lemma}

\allowdisplaybreaks

\begin{document}

\begin{frontmatter}



\title{Adaptive scaling for soft-thresholding estimator}


\author{Katsuyuki Hagiwara}

\ead{hagi@edu.mie-u.ac.jp}
\address{Faculty of Education, Mie University,\\ 1577 Kurima-Machiya-cho, Tsu, 514-8507, Japan}

\begin{abstract}
Soft-thresholding is a sparse modeling method that is typically applied
to wavelet denoising in statistical signal processing and analysis. It
has a single parameter that controls a threshold level on wavelet
coefficients and, simultaneously, amount of shrinkage for coefficients
of un-removed components. This parametrization is possible to cause
excess shrinkage, thus, estimation bias at a sparse representation;
i.e. there is a dilemma between sparsity and prediction accuracy. 
To relax this problem, we considered to introduce positive scaling on
soft-thresholding estimator, by which threshold level and amount of
shrinkage are independently controlled. Especially, in this paper, we
proposed component-wise and data-dependent scaling in a setting of
non-parametric orthogonal regression problem including discrete wavelet
transform. We call our scaling method adaptive scaling. We here employed
soft-thresholding method based on LARS(least angle regression), by which
the model selection problem reduces to the determination of the number
of un-removed components. We derived a risk under LARS-based
soft-thresholding with the proposed adaptive scaling and established a
model selection criterion as an unbiased estimate of the risk. We also
analyzed some properties of the risk curve and found that the model
selection criterion is possible to select a model with low risk and high
sparsity compared to a naive soft-thresholding method. This theoretical
speculation was verified by a simple numerical experiment and an
application to wavelet denoising.
\end{abstract}


\end{frontmatter}



\section{Introduction}

Orthogonal transform such as discrete wavelet transform is an important
tool in statistical signal processing and analysis. Especially, wavelet
denoising is a popular application of discrete wavelet transform. In
wavelet denoising, noisy signal is transformed into wavelet domain in
which wavelet coefficients are obtained. By applying a thresholding
method, noise-related parts of coefficients are removed in a sense;
e.g. some of coefficients are set to zero. The inverse wavelet transform
of the modified coefficients yields a denoised signal.  The most popular
and simple methods of thresholding is hard and soft-thresholding in
\citep{DJ1994,DJ1995}. Both thresholding methods have a parameter.  In
hard-thresholding method, the parameter works purely as a threshold
level; i.e. coefficients less than the parameter value are removed and
un-removed coefficients are harmless.  On the other hand, in
soft-thresholding, the parameter works as a threshold level as in
hard-thresholding and simultaneously as an amount of shrinkage for
un-removed components. Coefficients less than the parameter value are
removed and un-removed coefficients are shrunk toward zero by the
parameter.  For a better denoising performance, we need to determine an
optimal value of the parameter. For example, in hard-thresholding, if
the parameter value is too large then most of coefficients are removed
even when those are significant. This results in an excess smoothing
that yields a large bias between estimated output and target function
output. On the other hand, if the parameter value is too small then most
of coefficients are un-removed even when those are not significant. This
results in a large variance of output estimate and, thus useless for
denoising.  A problem of choice of an optimal parameter value is often
referred as a model selection problem. There are several model selection
methods under thresholding. \citep{DJ1994} has proposed universal hard
and soft-thresholding in which a theoretically significant constant
value is employed as a parameter value. Also, \citep{DJ1994} has derived
a criterion for determining an optimal parameter value of
soft-thresholding by applying Stein's lemma\citep{CS1981}. The
soft-thresholding method with this criterion is called as SURE (Stein's
Unbiased Risk Estimator) shrink in \citep{DJ1994}. Unfortunately, there
is no such a theoretically supported criterion for hard-thresholding
while modified cross validation approaches have been proposed
\citep{NGP1996,HT1998}.

We focus on a soft-thresholding method in this paper. As previously
mentioned, soft-thresholding is a combination of hard-thresholding and
shrinkage in which both of threshold level and amount of shrinkage are
simultaneously controlled by a single parameter. The parameter is a
threshold level for removing un-necessary components and is also an
amount of shift by which estimators of coefficients of un-removed
components are shrunk toward to zero. If the parameter value is large
then threshold level is large. Therefore, the number of un-removed
components is small. However, at the same time, the amount of shrinkage
is automatically large. This can be an excess shrinkage amount which may
yields a large bias of output estimate in representing a target
function.  This may cause a high prediction error at a relatively small
model even when it can represent a target function; i.e. even when it
can obtain a sparse representation.  Therefore, the number of un-removed
components in soft-thresholding tends to be large if we choose the
parameter value based on a substitution of prediction error such as SURE
and cross-validation error.  This is an inevitable problem of
soft-thresholding, which is brought about by an introduction of a single
parameter for controlling both of threshold level and amount of
shrinkage simultaneously. Note that, in the implementation of
thresholding methods for wavelet denoising in \cite{DJ1994},
thresholding is recommended to apply only to detail coefficients. This
heuristics may be actually valid to avoid the problem mentioned here.

On the other hand, in machine learning and statistics, there are several
model selection methods by using regularization, in which coefficient
estimators are obtained by minimizing a regularized cost that consists
of error term plus regularization term. A regularization method has a
parameter that is multiplied by regularizer in the regularization term
and determines a balance between error and regularization.  LASSO (Least
Absolute Shrinkage and Selection Operator) is a very popular
regularization method for variable selection\cite{LASSO}. It employs sum
of absolute values of coefficients as a regularizer; i.e. $\ell_1$ norm
of a coefficient vector. LASSO is known to be useful for obtaining a
sparse representation of a target function; i.e. the number of
components for representing a target function is very small. In LASSO,
extra components are automatically removed by setting their coefficients
to zero. This property is clearly understood when it applied to
orthogonal regression problems. In this case, LASSO reduces to a
soft-thresholding method in which a parameter of soft-thresholding is a
regularization parameter divided by 2. Hence, a sparseness obtained by
LASSO comes from a sof-thresholding property. And, thus, LASSO
encounters the above mentioned problem of soft-thresholding.  This
dilemma between sparsity and prediction of LASSO has already been
discussed in \citep{FL2001} and \citep{HZ2006}.  \citep{FL2001} has
proposed SCAD (Smoothly Clipped Absolute Deviation) penalty which is a
nonlinear modification of $\ell_1$ penalty.  \citep{HZ2006} has proposed
adaptive LASSO that employs weighted $\ell_1$ penalties.  An $\ell_1$
penalty term is modified by different ways (functions) in SCAD and
adaptive LASSO while shrinkage is suppressed for large values
of estimators in both methods. This may reduce an excess shrinkage at a
relatively small model.  Especially, in case of orthogonal regression,
weights of adaptive LASSO are effective for directly and adaptively
reducing a shrinkage amount that is represented as a shift in
soft-thresholding. In these methods, cross validation is used as a model
selection method for choosing parameter values such as a regularization
parameter. Unfortunately, usual cross validation can not be used in
orthogonal regression unless it is heuristically modified as in
\citep{NGP1996,HT1998}.

In this paper, we introduce a scaling of soft-thresholding estimators;
i.e. a soft-thresholding estimator is multiplied by a scaling
parameter. Unlike adaptive LASSO, introduction of scaling is intended to
control threshold level and amount of shrinkage independently.  It is
thus a direct solution for a problem of parametrization of
soft-thresholding. If the scaling parameter value is less than one then
it works as shrinkage of soft-thresholding estimator.  For an orthogonal
regression problem, this is equivalent to elastic net\cite{ELNET} in
machine learning. However, the scaling parameter can be larger than one
by which the above mentioned excess shrinkage in soft-thresholding is
expected to be relaxed; i.e. scaling expands a shrinkage estimator
obtained by soft-thresholding.  Especially in this paper, we propose a
component-wise and data-dependent scaling method; i.e. scaling parameter
value can be different for each coefficient and is calculated from data.
We refer the proposed scaling as adaptive scaling.  In this paper, we
derive a risk under adaptive scaling and construct a model selection
criterion as an unbiased risk estimate. Therefore, our work establishes
a denoising method in which a drawback of a naive soft-thresholding is
improved by the introduction of adaptive scaling and an optimal model is
automatically selected according to a derived criterion under the adaptive
scaling.

In Section 2, we state a setting of orthogonal non-parametric regression
that includes a problem of wavelet denoising. In this section, we also
give a naive soft-thresholding method and several related methods. In
this paper, especially, we employ a soft-thresholding method based on
LARS (least angle regression)\cite{LARS} in these methods.  In
LARS-based soft-thresholding, a model selection problem reduces to the
determination of the number of un-removed components. In Section 3, we
define an adaptive scaling and derive a risk under LARS-based
soft-thresholding with the adaptive scaling. We then give a model
selection criterion as an unbiased estimate of the risk. We here also
consider the properties of risk curve and reveals the model selection
property. The proofs of theorems in this section are included in
Appendix with some lemmas. In Section 4, the proposed adaptive scaling
method is examined for toy artificial problems including applications to
wavelet denoising. Section 5 is devoted to conclusions and future works.

\section{Non-parametric orthogonal regression}

\subsection{Setting and assumption of orthogonal non-parametric regression}

Let $\x=(x_1,\ldots,x_m)$ and $y$ be input variables and an output
variable, for which we have $n$ i.i.d. samples :
$\{(\x_i,y_i):i=1,\ldots,n\}$, where $\x_i=(x_{i,1},\ldots,x_{i,m})$.
We assume that $y_i=h(\x_i)+e_i$, $i=1,\ldots,n$, where $e_1,\ldots,e_n$
are i.i.d additive noise sequence according to $N(0,\sigma^2)$;
i.e. normal distribution with mean $0$ and variance $\sigma^2$. $h$ is a target
function. We assume that $\x_1,\ldots,\x_n$ are fixed below. We define
$\y=(y_1,\ldots,y_n)'$, $\h=(h(\x_1),\ldots,h(\x_n))'$ and
$\e=(e_1,\ldots,e_n)'$, where $'$ denotes a matrix transpose.  We then
have $\y=\h+\e$ and $\E[\y]=\h$, where $\E$ denotes the
expectation with respect to the joint probability distribution of $\y$.

Let $g_1,g_2,\ldots$ be a series of functions on $\R^m$. We consider to
estimate a target function by a linear combination of $n$ 
functions in this series :
\begin{equation}
f_{\b}(\x)=\sum_{j=1}^nb_jg_j(\x),~\x\in\R^m,
\end{equation}
where $\b=(b_1,\ldots,b_n)'$ is a coefficient vector. This is a
non-parametric regression problem.  We call $g_j$ a component or basis
function.  We assume that there exists $n^*$ and
$\vbeta=(\beta_1,\ldots,\beta_n)'$ such that
$h(\x)=\sum_{j=1}^{n}\beta_jg_j(\x)$ for any $\x\in\R^m$ when $n\ge
n^*$. $\beta_j$ can be zero for some $j$. 
We define $K^*=\{j:1\le j\le
n,~\beta_j\neq 0\}$ and denote the complement of $K^*$ by $\oK^*$.  We
call $g_j$ with $j\in K^*$ true component or non-zero component.  We
also define $k^*=|K^*|$ which is the number of true components or
non-zero components. 
We assume that $n\ge n^*$; i.e. true components are always included in a model.
We also assume that $k^*$ is very small compared to $n$. These two assumptions
say that there exists a sparse representation of a target function in
terms of a set of $n$ components.

Let $\G$ be an $n\times n$ matrix whose $(i,j)$ element is
$g_j(\x_i)$. We assume that the orthogonality condition :
\begin{equation}
\label{eq:ot-cond}
\G'\G=n\I_n,  
\end{equation}
where $\I_n$ denotes an $n\times n$ identity matrix. We thus consider a
non-parametric orthogonal regression problem; e.g. discrete Fourier
transform and discrete wavelet transform for typical examples. The
least squares estimator under the orthogonality condition is given by
\begin{equation}
\label{eq:ecv}
\ecv=(\ec_1,\ldots,\ec_n)'=\frac{1}{n}\G'\y.
\end{equation}
Note that we have $\y=\G\ecv$ here. 
Since there exists a $\vbeta$ such that $\h=\G\vbeta$,
\begin{equation}
\label{eq:dist-of-ecv}
\ecv\sim N\left(\vbeta,\frac{\sigma^2}{n}\I_n\right)
\end{equation}
holds by the assumption on additive noise; i.e. multivariate normal
distribution with a mean vector $\vbeta$ and a unit covariance matrix
multiplied by $\sigma^2/n$. In other words, $\ec_j\sim
N(\beta_j,\sigma^2/n)$, $j=1,\ldots,n$ and $\ec_1,\ldots,\ec_n$ are
independent. We define $s_j=\sign(\ec_j)$, $j=1,\ldots,n$, where $\sign$
is a sign function. We define $p_1,\ldots,p_n$ as an index sequence for
which $|\ec_{p_1}|\ge\cdots\ge|\ec_{p_n}|$ holds. Note that we can
exclude the case of ties in our probabilistic evaluations in this paper
since this is guaranteed with probability one by (\ref{eq:dist-of-ecv}).

\subsection{LASSO, LARS, elastic net and adaptive LASSO}

Let $\ebv_{\theta}=(\eb_{\theta,1},\ldots,\eb_{\theta,n})$ with a
parameter $\theta\ge 0$ be a soft-thresholding estimator, in which
\begin{equation}
\label{eq:st-estimator}
\eb_{\theta,j}=(\ec_j-\theta)_{+}s_j,~j=1,\ldots,n
\end{equation}
where $(u)_+=\max(u,0)$. $\theta$ determines both of a threshold level
and amount of shrinkage. Under the orthogonality condition, several
sparse modeling methods can be reduced to soft-thresholding estimator.

For a fixed $\lambda_1\ge 0$, cost function of LASSO is given by
\begin{equation}
\label{eq:cost-lasso}
S_{\lambda_1}(\b)=\frac{1}{n}\|\y-\G\b\|^2+\lambda_1\|\b\|_1,
\end{equation}
where $\|\cdot\|$ is the Euclidean norm and
$\|\b\|_1=\sum_{k=1}^n|b_j|$; i.e. LASSO introduces an $\ell_1$
regularizer. $\lambda_1$ is a regularization parameter.  A minimizer of
(\ref{eq:cost-lasso}) under the orthogonality condition is known to be a
soft-thresholding estimator with $\theta=\lambda_1/2$; i.e. it is
$\ebv_{\lambda_1/2}$.  On the other hand, for fixed $\lambda_1\ge 0$ and
$\lambda_2\ge 0$, cost function of elastic net is given by
\begin{equation}
\label{eq:cost-elnet}
S_{\lambda_1,\lambda_2}(\b)=\frac{1}{n}\|\y-\G\b\|^2+\lambda_1\|\b\|_1+\lambda_2\|\b\|^2.
\end{equation}
Thus, elastic net introduces both of an $\ell_1$ regularizer and 
$\ell_2$ regularizer.  As shown in \citep{ELNET}, a minimizer of
(\ref{eq:cost-elnet}) under the orthonormality condition is given by
$\eb_{\lambda_1/2,k}/(1+\lambda_2)$, $k=1,\ldots,n$.  Since
$\lambda_2\ge 0$, the solution of elastic net is obtained by shrinking
LASSO estimator which is a soft-thresholding estimator.

On the other hand, LARS (Least angle regression) \cite{LARS} is a greedy
iterative algorithm in which a component is appended to a model at each
step. This can be viewed as a sparse modeling method if we can find an
optimal step.  For this purpose, a $C_p$ type criterion is derived under
a mild condition in \citep{LARS}.  As shown in \citep{KH2014} and Lemma
1 in \citep{LARS}, under the orthonormality condition, LARS is also
reduced to soft-thresholding estimator in which the parameter value is
given by $\theta=|\ec_{p_{k+1}}|$ at the $k$th step; i.e. it is the
$(k+1)$th largest absolute value among the least squares estimators.
Therefore, a set of candidates of parameter values is
$\{|\ec_{p_1}|,\ldots,|\ec_{p_n}|\}$ in LARS.  By this choice of
threshold level, the number of un-removed components at the $k$th step
is equal to $k$.  Therefore, a model selection problem of LARS-based
soft-thresholding is the
determination of the number of un-removed components. We refer to
LARS-based soft-thresholding as LST.

As a modification of LASSO, adaptive LASSO\citep{HZ2006} introduces a
weighted $\ell_1$ regularizer, in which a weight for the $j$th component
is $w_j$ and a choice of $w_j=1/|\ec_j|^{\gamma}$ with $\gamma>0$ is especially
considered in \citep{HZ2006}.  The solution of adaptive LASSO under the
orthonormality condition is given by
\begin{equation}
\label{eq:o-alasso-est}
\eb_{w_j,\lambda_1,j}=(|\ec_j|-w_j\lambda_1/2)_+s_j,~j=1,\ldots,n.
\end{equation}
It is regarded as a soft-thresholding estimator with a component-wise
and data-dependent parameter. If $|\ec_j|$ is large then $w_j$ is
small. In this case, threshold level and amount of shrinkage for the
corresponding estimator is small. This reduces a bias, or equivalently,
an excess shrinkage of estimator especially when the estimator is
actually valid; i.e. the corresponding component is needed. In other
words, adaptive LASSO avoids an excess shrinkage on estimators of
un-removed components by an adaptive manner; i.e. by controlling a
component-wise and estimator-dependent ``shift'' in soft-thresholding
estimator. This relaxes the problem of employing a single parameter value
for both of threshold level and amount of shrinkage in soft-thresholding. We can
choose a small parameter value for valid components and a large value for
non-essential components; i.e. the parameters
mainly work as threshold levels for removing non-essential components.

In this paper, by introducing scaling for soft-thresholding estimator,
we consider to control threshold level and amount of shrinkage
independently. Our approach is different from adaptive LASSO while they
serves the same purpose. As seen in later sections, the advantage of
employing scaling is that we can construct a model selection criterion
that is required in applications.

\section{Adaptive scaling}

\subsection{Component-wise scaling and some special cases}

Let $\ebv_k=(\eb_{k,1},\ldots,\eb_{k,n})$ be a vector of the above mentioned
LST estimators that are defined by
\begin{equation}
\label{eq:def-eb-k-j}
\eb_{k,j}=(|\ec_j|-\etheta_k)_{+}s_j,~j=1,\ldots,n,
\end{equation}
where $\etheta_k=|\ec_{p_{k+1}}|$.  We define
$\valpha=(\alpha_1,\ldots,\alpha_n)$ for $0<\alpha_j<\infty$.  In this
paper, we consider to employ
$\ebv_{k,\valpha}=(\eb_{k,1,\alpha_1},\ldots,\eb_{k,n,\alpha_n})$ in
which 
\begin{equation}
\label{eq:def-eb-a-k-j}
\eb_{k,j,\alpha_j}=\alpha_j\eb_{k,j,\alpha_j},~j=1,\ldots,n.
\end{equation}
We call $\alpha_j$, $j=1,\ldots,n$ component-wise scaling parameters.
Let $\A$ be an $n\times n$ diagonal matrix whose $(j,j)$ element is
$\alpha_j$. We can write $\ebv_{k,\valpha}=\A\ebv_k$.  We define
$\evmu_{k,\valpha}=\G\ebv_{k,\valpha}=\G\A\ebv_k$.  Note that, in a
matrix formulation, $\ebv_{k,\valpha}$ and $\ebv_k$ are used as vertical
vectors. As in the previous discussion, if we restrict
$\alpha_j=\alpha\le 1$ then the method is elastic net which yields
shrinkage of soft-thresholding estimator. Therefore, introduction of
scaling parameter can be viewed as an extension of elastic net. However,
we expect that scaling is used for expanding soft-thresholding
estimator; i.e. $\alpha_j>1$ is desirable. Note that $\evmu_{k,\valpha}$
is a two stage estimate in which LST is firstly applied and then scaling
is applied. Scaling re-adjusts only amount of shrinkage. A risk for LST
with component-wise scaling is defined by
\begin{equation}
\label{eq:risk-for-scaling-ST-def}
R_{n,k}(\valpha)=\frac{1}{n}\E\|\evmu_{k,\valpha}-\h\|^2=\E\|\A\ebv_k-\vbeta\|^2.
\end{equation}
where the latter definition is due to the orthogonality condition (\ref{eq:ot-cond}). A
naive LST is a case of $\valpha=\1_n$, where
$\1_n$ is an $n$-dimensional vector of one's. For this case, we have
\begin{equation}
\label{eq:risk-for-LST}
R_{n,k}(\1_n)=\frac{1}{n}\E\|\evmu_{k,\1_n}-\y\|^2-\sigma^2+\frac{2\sigma^2}{n}
\end{equation}
as a special case of \cite{LARS}.
More generally, in case of introducing a single common scaling parameter $\alpha$ on all
components, \citep{KH2015} has shown that
\begin{equation}
\label{eq:risk-for-LST-SSP}
R_{n,k}(\alpha\1_n)
=\frac{1}{n}\E\|\evmu_{k,\alpha\1_n}-\y\|^2-\sigma^2+\frac{2\sigma^2\alpha}{n}.
\end{equation}
Therefore, an unbiased risk estimate is given by
\begin{equation}
\label{eq:unbiased-risk-estimate-for-LST-SSP}
\eR_{n,k}(\alpha\1_n)
=\frac{1}{n}\|\evmu_{k,\alpha\1_n}-\y\|^2-\sigma^2+\frac{2\sigma^2\alpha}{n}
\end{equation}
which can be used as a model selection criterion for choosing an
optimal $k$ if we replace $\sigma^2$ with its estimate $\esigma^2$.  For
this case, an optimal scaling value that minimizes the risk is given by
\begin{equation}
\alpha_{\rm opt}=\frac{\E\left[\sum_{j\in\eK_k}\eb_{k,j}\ec_j\right]+\sigma^2k/n}
{\E\left[\sum_{j\in\eK_k}\eb_{k,j}^2\right]}.
\end{equation}
In practical application, for example, 
\begin{equation}
\label{eq:lst-ssp-ealpha}
\ealpha=\frac{\sum_{j\in\eK_k}\eb_{k,j}\tc_j+\esigma^2k/n}{\sum_{j\in\eK_k}\eb_{k,j}^2}
\end{equation}
can be an estimate of the optimal value.

\subsection{Definitions for theorems and lemmas}

We state some definitions used below.  We define
$c_j=\sqrt{n}\beta_j/\sigma$ and $\c=(c_1,\ldots,c_n)$.  We define
$\tc_j=\sqrt{n}\ec_j/\sigma$ and $\tcv=(\tc_1,\ldots,\tc_n)$, by which
$\tcv\sim N(\c,\I_n)$ ; i.e.  $\tc_j\sim N(\sqrt{n}\beta_j,1)$ and
$\tc_1,\ldots,\tc_n$ are mutually independent. We define
$\ttheta_k=\sqrt{n}\etheta_k/\sigma=|\tc_{p_{k+1}}|$ in applying LST. Correspondingly, by
(\ref{eq:def-eb-k-j}), we define
\begin{equation}
\label{eq:def-tb-k-j}
\tb_{k,j}=\sqrt{n}\eb_{k,j}/\sigma=(|\tc_j|-\ttheta_k)s_j,~j=1,\ldots,n
\end{equation}
and $\tbv_k=(\tb_{k,1},\ldots,\tb_{k,n})$.
We also define
$\oc_i=\tc_i-\sqrt{n}\beta_j/\sigma$, by which $\oc_1,\ldots,\oc_n$ are
i.i.d. according to $N(0,1)$ by the definition of
$\tc_1,\ldots,\tc_n$. 
For an event $E$, we denote the complement of $E$ by $\oE$ and indicator
function of $E$ by $I_{E}$.  
We define $E_{n,l}^*=\{p_l\in K^*\}$ 
and $E_n^*=\bigcap_{l=1}^{k^*}E_{n,l}^*$. 
We also define
$F_j=\{\tc_j^2\le\max_{i\in\oK^*}\tc_i^2\}$. 
We denote
$\chi^2$ distribution with one degree of freedom by $\chi^2_1$.

\subsection{Definition of adaptive scaling}

The purpose of scaling is to avoid excess shrinkage of coefficients
of un-removed components. Then, it is reasonable to choose $\alpha_j$ so as
to satisfy $\alpha_j\eb_{k,j}=\ec_j$. This yields
\begin{equation}
\label{eq:approx-in-adaptive-scaling}
\alpha_j=1/\left(1-\etheta_k/|\ec_j|\right)\simeq 1+\etheta_k/|\ec_j|,~j=1,\ldots,n
\end{equation}
when $\etheta_k/|\ec_j|$ is small.  This approximation is valid since an
un-removed component may have a coefficient estimate that is
sufficiently larger than an appropriate threshold level.  In this paper,
we hence employ
\begin{equation}
\label{eq:component-wise-scaling}
\ealpha_j=
\begin{cases}
1+\etheta_k/|\ec_j|=1+\ttheta_k/|\tc_j| & \ec_j\neq 0\\
\alpha & \ec_j=0
\end{cases}
,~j=1,\ldots,n  
\end{equation}
as empirical values, where $\alpha$ is a finite constant that is defined
to avoid $\ealpha_j=\infty$ when $\ec_j=0$. We define
$\evalpha=(\ealpha_1,\ldots,\ealpha_n)$.
(\ref{eq:component-wise-scaling}) gives data-dependent and
component-wise scaling value. We refer this scaling method as adaptive
scaling.  By (\ref{eq:component-wise-scaling}), the adaptive scaling
value is always larger than one.  Note also that $\ealpha_j$ is valid
only to $j\in\eK_k$ since $\tb_{k,j}=0$ for $j\notin\eK_k$.
Let $\eAv$ be an $n\times n$ diagonal matrix whose $(j,j)$ element is $\ealpha_j$.
We define a risk for our adaptive scaling estimator by
\begin{equation}
\label{eq:RAS-def}
\RAS(n,k)=\E\|\eAv\ebv_k-\vbeta\|^2.
\end{equation}

\subsection{Main results}

We state three theorems whose proofs are given in Appendix with
some lemmas.

\begin{theorem}
\label{theorem:Rnk-evalpha}
For $\evalpha$ defined in (\ref{eq:component-wise-scaling}),
\begin{equation}
\label{eq:theorem-Rnk-evalpha}
\RAS(n,k)=\frac{1}{n}\E\|\evmu_{k,\evalpha}-\y\|^2-\sigma^2
+\frac{2\sigma^2k}{n}
+\frac{2\sigma^2}{n}\E\left[\sum_{j\in\eK_k}(\ealpha_j-1)^2\right]
\end{equation}
holds. 
\end{theorem}

\begin{theorem}
\label{theorem:ealpha_j-convergence}
We define
\begin{equation}
\label{eq:def-epsilon_n}
\epsilon_{j,n}=\epsilon_{j,n}(\delta)=\frac{1}{|\beta_j|+\delta}\sqrt{\frac{2\log n}{n}}
\end{equation}
with $\delta>0$.
For $j\in K^*$,
\begin{equation}
\label{eq:ealpha_j-lower-bound-K*-0}
\lim_{n\to\infty}\P\left[|\ealpha_j-1|>\epsilon_{j,n}\right]=0
\end{equation}
holds. This implies that, for $j\in K^*$,
\begin{equation}
\label{eq:ealpha_j-lower-bound-K*}
\lim_{n\to\infty}\P\left[|\ealpha_j-1|>\epsilon\right]=0
\end{equation}
holds for any $\epsilon>0$. On the other hand, 
we assume that $k> k^*$. Then, for $j\in\oK^*$, 
\begin{equation}
\label{eq:ealpha_j-lower-bound-oK*}
\lim_{n\to\infty}\P\left[\ealpha_j<2-\epsilon\right]=0
\end{equation}
holds for any $\epsilon>0$.
\end{theorem}

\begin{theorem}
\label{theorem:R(1)-R(evalpha)}
\begin{equation}
\label{eq:theorem-R(1)-R(evalpha)}
R_{n,k^*}(\1_n)-\RAS(n,k^*)\ge 2\sigma^2k^*\frac{\log n}{n}
\end{equation}
holds for a sufficiently large $n$.
\end{theorem}

We give some remarks.

\begin{itemize}

\item By Theorem \ref{theorem:Rnk-evalpha},
\begin{align}
\label{eq:eRnk_evalpha}
\eR_{n,k}(\evalpha)=\frac{1}{n}\|\evmu_{k,\alpha}-\y\|^2-\sigma^2+\frac{2k\sigma^2}{n}
+\frac{2\sigma^2}{n}\sum_{j\in\eK_k}(\ealpha_j-1)^2
\end{align}
is an unbiased estimator of risk under adaptive scaling with
$\evalpha$ defined by (\ref{eq:component-wise-scaling}). Therefore, this
      can be a model selection criterion for choosing an
      optimal $k$ if we can set an appropriate estimate of noise
      varinace $\sigma^2$ in (\ref{eq:eRnk_evalpha}).

\item By Lemma \ref{lemma:P-oEn*-bound}, the probability that all of
      true components are un-removed is high when $k\ge k^*$ and $n$ is
      sufficiently large; i.e. LST has a kind of consistency in
      selecting true components if those exist. Note that since our
      adaptive scaling is applied to LST estimator, this consistency
      result applies to adaptive scaling estimators.

\item Theorem \ref{theorem:ealpha_j-convergence} says that, in a large
      sample situation, scaling values are larger than $2$ for
      components that are not true. Some of non-true components are
      selected when $k>k^*$. This excess expansion of coefficient
      estimators for non-true components may cause a high risk for
      $k>k^*$. Therefore, $\RAS(n,k)>R_{n,k}(\1_n)$ may hold for
      $k>k^*$ even though $\RAS(n,k^*)<R_{n,k^*}(\1_n)$ holds
by Theorem \ref{theorem:R(1)-R(evalpha)}. This fact seems to be
      disadvantage of introducing our adaptive scaling. However, it may
      not be so from the viewpoint of model selection since this
      property allows us to identify the minimum of risk curve; i.e. risk is small
      at $k=k^*$ while it is large when $k\neq k^*$.
Hence, nearly optimal $k$ is expected to be found
according to a model selection criterion given by (\ref{eq:eRnk_evalpha}).
And, at such an optimal $k$, a consistent choice of a set of true
      components and a low risk value are guaranteed by Lemma
      \ref{lemma:P-oEn*-bound} and Theorem \ref{theorem:R(1)-R(evalpha)}
      respectively. This speculation is verified in numerical
      experiments in the next section. 

\item Since the least squares estimators of coefficients of true
components tend to be large, approximation in
(\ref{eq:approx-in-adaptive-scaling}) is valid for them. Therefore, LST
estimators for true components are nearly the least squares
      estimators. And, as mentioned above, true components may be
      consistently selected according to (\ref{eq:eRnk_evalpha}) if
      variance estimate is suitable. Therefore, a model estimated by our
      adaptive scaling scheme may be close to one estimated by a hard
      thresholding for which it is difficult to establish a model selection procedure.

\end{itemize}

\section{Numerical experiments}

\subsection{Case of known true components}

Consider a set of $n$ functions
$G_n=\left\{g_1,g_2,\ldots,g_n\right\}$ in which
\begin{align}
g_k(x)&=
\begin{cases}
1 & k=1\\
\sqrt{2}\cos(k x/2) & \mbox{$k$ : even and $k\neq 1,n$}\\
\sqrt{2}\sin(k x/2).& \mbox{$k$ : odd and $k\neq 1,n$}\\
\cos(k x/2) & k=n\\
\end{cases}.
\end{align}
The design matrix constructed by $G_n$ satisfies the orthogonality
condition of (\ref{eq:ot-cond}) if $x_i=2\pi (i-1)/n$ for $i=1,\ldots,n$
and $n$ is even. These two conditions are satisfied in our experiment
here. We set $K^*=\{2,4,6,8\}$ and
$(\beta_2,\beta_4,\beta_6,\beta_8)=(2.0,-1.5,1.0,-0.5)$, by which
$h(x_i)=\sum_{k\in K^*}\beta_kg_k(x_i)$; i.e.  $g_k$, $k\in K^*$ are
true components.  We set $\sigma^2=1$ for Gaussian noise variance. We
also set $n=500$ and the maximum number of components that is included
in a model is $50$. For an artificially generated data, we apply LST,
LST-SSP(LST with single scaling parameter) and LST-AS(LST with an
adaptive scaling). We employ (\ref{eq:lst-ssp-ealpha}) as an empirical
scaling value for LST-SSP. Adaptive scaling values for LST-AS are given
by (\ref{eq:component-wise-scaling}). For each method, we calculate the
approximated risk that is the mean-squared error between a true function
output and estimated output on data points.  We also calculate the risk
estimate (unbiased estimate of risk). It is given by
(\ref{eq:risk-for-LST}) with $\alpha=1$ for LST,
(\ref{eq:risk-for-LST-SSP}) with $\ealpha$ in (\ref{eq:lst-ssp-ealpha})
for LST-SSP and (\ref{eq:eRnk_evalpha}) with $\ealpha_j$ in
(\ref{eq:component-wise-scaling}) for LST-AS.  We need to estimate noise
variance in calculating a risk estimate that is employed as a model
selection criterion in applications. We here estimate it by the unbiased
estimate of noise variance under a linear regression with a set of $250$
components that includes true components. We repeat this procedure
for $1000$ times.

\begin{figure}[p]
\begin{center}
\includegraphics[width=90mm]{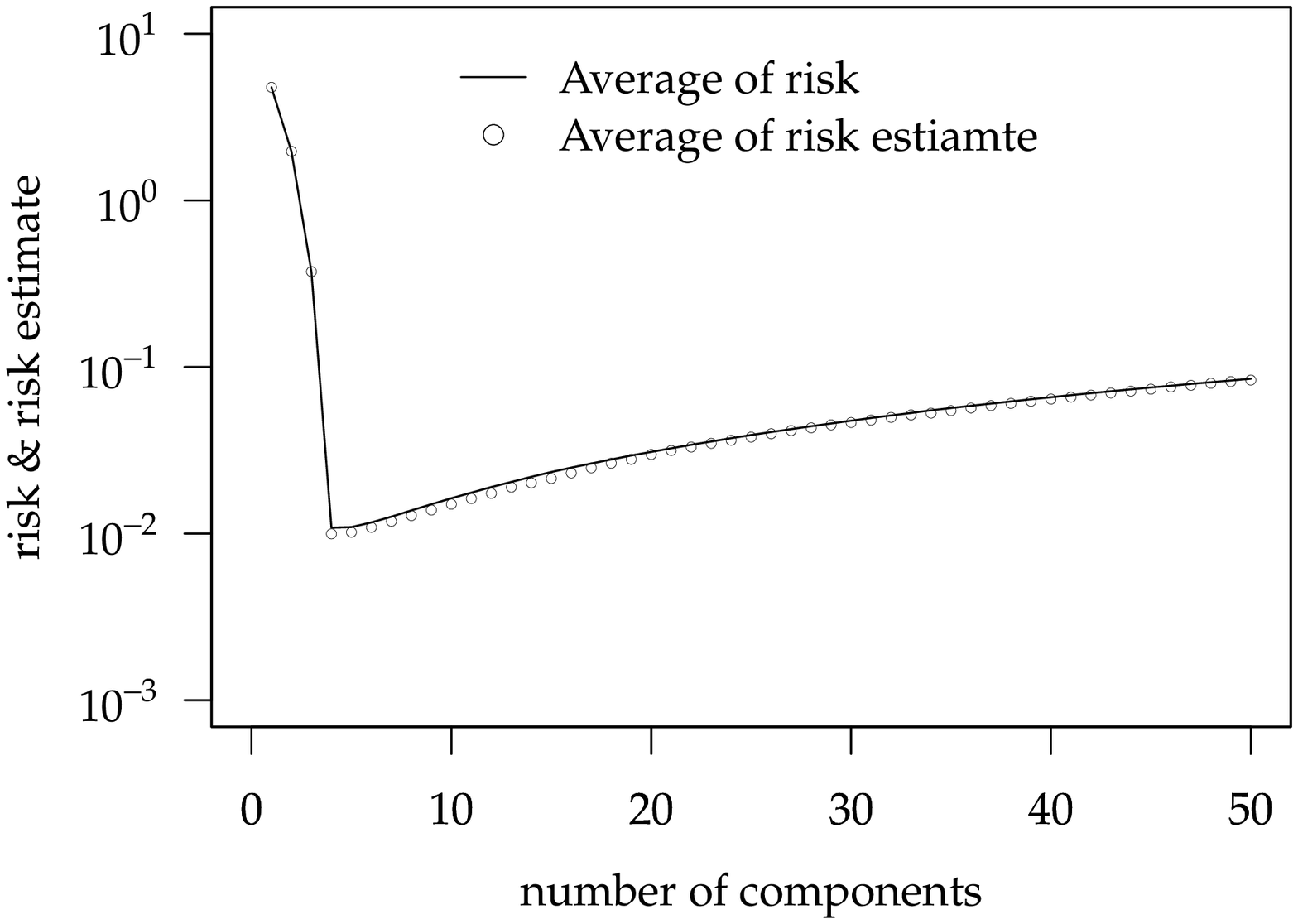} 
\caption{Averages of risk and risk estimate for LST-AS.}
\label{fig:risk-and-risk-estimate-LST-AS}
\end{center}
\begin{center}
\includegraphics[width=90mm]{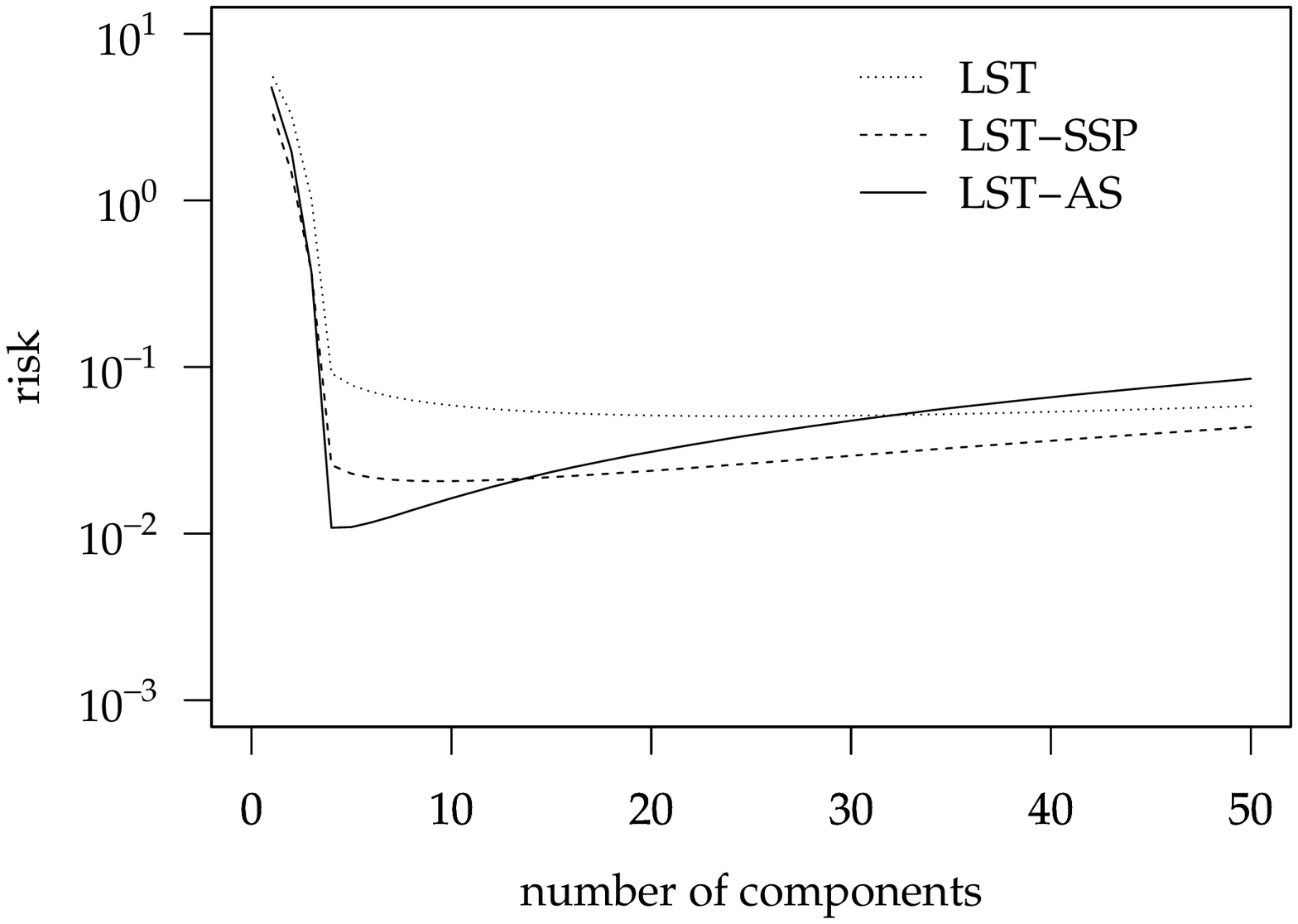} 
\caption{Averages of risk for LST, LST-SSP and LST-AS.}
\label{fig:risk-for-LST-SSP-AS}
\end{center}
\end{figure}

We show averages of (approximated) risks and risk estimates for LST-AS
in Figure \ref{fig:risk-and-risk-estimate-LST-AS}. We also show averages
of risks for LST, LST-SSP and LST-AS in Figure
\ref{fig:risk-for-LST-SSP-AS}. In Figure
\ref{fig:risk-and-risk-estimate-LST-AS}, we can see that
(\ref{eq:eRnk_evalpha}) is actually valid as an unbiased risk estimate
under LST-AS even when noise variance is replaced with its estimator. In
Figure \ref{fig:risk-for-LST-SSP-AS}, at around a small number of
components, risk of LST-AS is minimized and is smaller than those of LST
and LST-SSP. However, risk of LST-AS tends to be larger than those of
LST and LST-SSP as the number of components increases.  This is
consistent with the remark on Theorem \ref{theorem:R(1)-R(evalpha)} and
Theorem \ref{theorem:ealpha_j-convergence}. In other words, an optimal
number of components can clearly be identified in risk curve of LST-AS
while risk curves of LST and LST-SSP are nearly flat around the minimum
value in Figure \ref{fig:risk-for-LST-SSP-AS}. We emphasize two
important points in this result. The first one is that, as guaranteed by
Theorem \ref{theorem:R(1)-R(evalpha)}, risk value of LST-AS is smaller
than that of LST at around an optimal number of components. The second
one is that it can be found via a model selection based on risk
estimate. In Table \ref{tbl:risk-and-un-removed-components}, we show the
averaged risk value and the average numbers of un-removed components
for models that are selected by risk estimates.
From Table \ref{tbl:risk-and-un-removed-components}, we can say that
LST-AS gives low risk at a sparse representation.

\begin{table}[h]
\begin{center}
\caption{Average of risk and the number of un-removed components 
\label{tbl:risk-and-un-removed-components} 
selected according to risk estimate. The standard deviation
is showed in the bracket.}

\vspace{2mm}

\begin{tabular}{|c|c|c|}\hline
method & risk & \#un-removed components\\\hline
LST & 0.0546~(0.0186) & 26.43~(12.99)\\\hline
LST-SSP & 0.0268~(0.0147) & 16.83~(12.5)\\\hline
LST-AS & 0.0232~(0.0202) & 10.18~(8.21) \\\hline
\end{tabular}
\end{center}
\end{table}

\subsection{Application to wavelet denoising}

Discrete wavelet transform is a popular tool for analysis, de-noising
and compression of signals and images; e.g. see \citep{BGG1998}. We here
consider an application of LST with adaptive scaling
to a problem of wavelet denoising\citep{DJ1994,DJ1995}. Let $y(t)$ ,
$t\in[0,1]$ be a signal. $n$ samples of $y(t)$ is denoted by
$y_i=y(t_i)$, $t_i=(i-1)/(n-1)$, $i=1,\ldots,n$. We define
$\y=(y_1,\ldots,y_n)$.  We assume that $n=2^J$ for a natural number $J$.
Let $\c_j=(c_{j,1},\ldots,c_{j,n_j})$ and
$\d_j=(d_{j,1},\ldots,d_{j,n_j},\ldots,d_{J-1,1},\ldots,d_{J-1,n_{J-1}})$
be approximation and detail coefficients at a level $J$ in discrete
wavelet transform, where $n_j=2^j$. We define $\w_j=(\c_j,\d_j)$ in
which we set $\w_J=\c_J$ for $j=J$. By setting $\c_J=\y$, the
decomposition algorithm with pre-determined wavelets calculates
$\w_{j-1}$ from $\c_j$ by decreasing $j=J,J-1,\ldots,J_0$, where $J_0$
is a fixed level determined by user. This procedure can be written by
\begin{equation}
\w_{J_0}=H_{J_0}\y, 
\end{equation}
where $H_{J_0}$ is an $n\times n$ orthonormal matrix that is determined
by coefficients of scaling and wavelet function; e.g. see
\citep{DJ1994,DJ1995}.  On the other hand, the reconstruction algorithm
calculates $\w_{j+1}$ from $\w_{j}$ by increasing
$j=J_0,\ldots,J-1,J$. This can be written by
\begin{equation}
\w_J=H_{J_0}'\w_{J_0}
\end{equation}
since $H_{J_0}$ is an orthonormal matrix.  Let $\Theta$ be an operator
on $\R^n$ into $\R^n$ such as a thresholding operator. In wavelet
denoising, $\w_{J_0}$ is processed by using $\Theta$ and obtain
$\owv_{J_0}=\Theta(\w_{J_0})$. We then obtain a denoised signal by
$\owv_J=H_{J_0}'\owv_{J_0}$. Note that, in applications, a simple and fast
decomposition/reconstruction algorithm is used instead of the above matrix
calculation; e.g. see \citep{BGG1998}.

We here compare the prediction accuracy and sparseness of LST-AS to those of LST,
LST-SSP and also universal soft-thresholding (UST) in \citep{DJ1994}.
Note that SURE shrink of \citep{DJ1995} is almost equivalent to
LST here.  In an application of UST, a
threshold level on the absolute values of coefficients at the $J_0$th
level is given by
\begin{equation}
\etheta_n=\sqrt{2\esigma^2\log n}, 
\end{equation}
where $\esigma^2$ is an estimate of noise variance.  In wavelet
denoising, the median absolute deviation (MAD) is a standard robust
estimate of noise variance. It is given by 
\begin{equation}
\esigma=\median\{|d_{J-1,1}|,\ldots,|d_{J-1,n_{J-1}}|\}/0.6745,
\end{equation}
where $d_{J-1,j}$, $j=1,\ldots,n_{J-1}$ is the smallest scale wavelet
coefficients that are heuristically known to be noise dominated
components. For LST, LST-SSP and LST-AS, we also employ this estimator
in a model selection criterion that is an unbiased risk estimate.

We choose ``heavisine'' and ``blocks'' given in \citep{DJ1994} as test
signals. The former is almost smooth and the latter has many
discontinuous points. Additive noise has a normal distribution with mean
$0$ and variance $\sigma^2=1$.  As in \citep{DJ1994}, signals are
rescaled so that signal-to-noise ratio is $7$.  The number of samples is
$n=1024$. We set $J_0=2$. In \citep{DJ1994}, in practical application,
it is employed a heuristic method which applies soft-thresholding only
for detail coefficients at a determined level. We do not obey this
heuristics and apply soft-thresholding to all coefficients in orthogonal
transformation for a fair comparison. This is because the choice of
a level at which thresholding applies largely depends on the performance
as in \citep{AY1996} and there is no systematic choice of such level. We
employ the orthogonal Daubechies wavelet with $8$ wavelet/scaling
coefficients.  For given samples, we apply LST, LST-SSP, LST-AS and UST,
in which the maximum number of un-removed components is set to $300$;
i.e. the maximum value of $k$ to be examined. We then calculate the mean
squared error between true signal outputs and estimated outputs on the
sampling points as an approximation of risk.  For LST, LST-SSP and
LST-AS, the mean squared error and risk estimate are obtained at each
$k$.  For UST, the number of un-removed components and risk value at a
selected size are obtained. We repeat this procedure $500$ times.

We show averages of (approximated) risk and risk estimate of LST, LST-SSP and
LST-AS in Figure \ref{fig:wl-risk-curve-heavisine} for ``heavisine''
and Figure \ref{fig:wl-risk-curve-blocks} for ``blocks'' respectively. We
also show box plots of risk values at the selected number of components
and those of the number of un-removed components in Figure
\ref{fig:wl-box-plot-heavisine} for ``heavisine'' and Figure
\ref{fig:wl-box-plot-blocks} for ``blocks'' respectively.  By Figure
\ref{fig:wl-risk-curve-heavisine} (b) and Figure
\ref{fig:wl-risk-curve-blocks} (b), risk estimate approximates risk well
for both signals even when noise variance is estimated by MAD.  By
Figure \ref{fig:wl-risk-curve-heavisine} and Figure
\ref{fig:wl-risk-curve-blocks}, we can expect that a model estimated by
LST-AS shows a low risk and high sparsity compared to LST and LST-SSP;
i.e. this result leads to the same conclusions as in the previous
numerical example.  By comparing Figure
\ref{fig:wl-risk-curve-heavisine} to Figure
\ref{fig:wl-risk-curve-blocks}, the optimal number of components for
``blocks'' is larger than for ``heavisine'', which is due to a degree of
smoothness of signals.  By Figure \ref{fig:wl-box-plot-heavisine} and
Figure \ref{fig:wl-box-plot-blocks}, for both signals, LST-AS
outperforms the other methods in terms of prediction accuracy and sparsity, in which
especially it shows a nice sparseness property. 
Note that the worse results of LST and UST may be improved
by applying a heuristics that thresholding methods are applied only
to detail coefficients at a determined level while there is no
systematic choice of the appropriate level.

\begin{figure}[p]
\begin{center}
\small
\includegraphics[width=90mm]{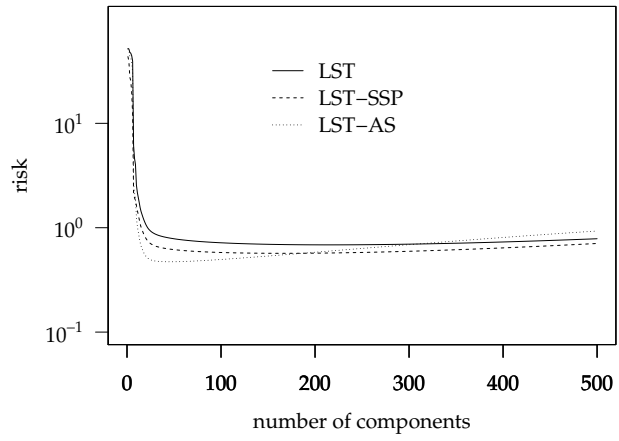} 

(a) Averaged risk curves of LST, LST-SSP and LST-AS.

\vspace{5mm}

\includegraphics[width=90mm]{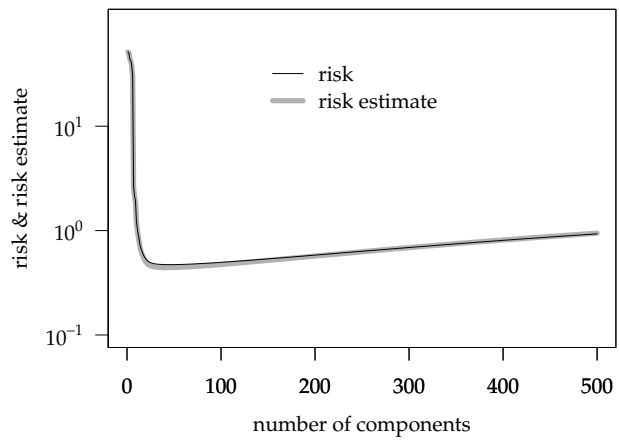} 

(b) Averaged risk and risk estimate of LST-AS.

\end{center}
\caption{Risk curve and risk estimate for ``heavisine''.}
\label{fig:wl-risk-curve-heavisine}
\end{figure}

\begin{figure}[p]
\begin{center}
\small
\includegraphics[width=90mm]{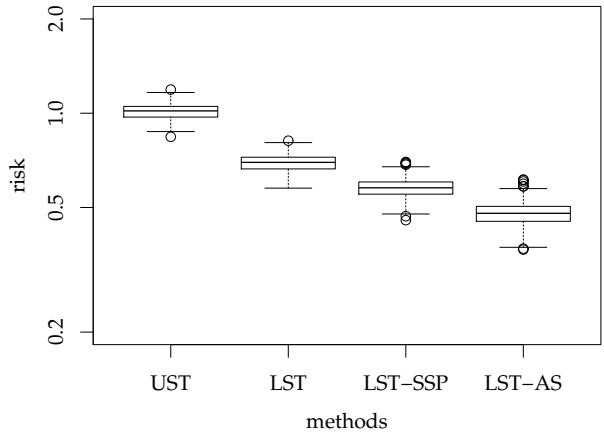} 

(a) Risk value at the selected number of components.

\vspace{5mm}

\includegraphics[width=90mm]{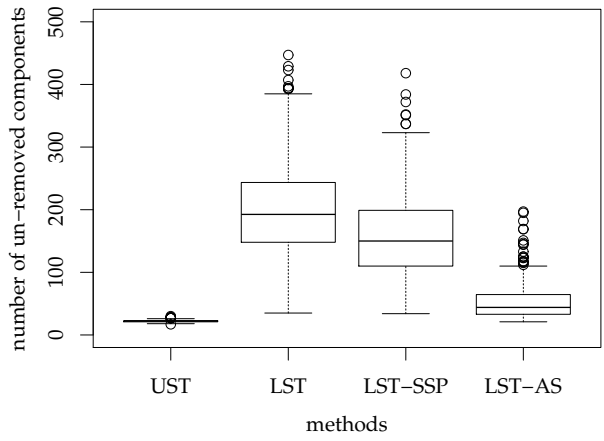} 

(b) The number of un-removed components.

\end{center}
\caption{Risk and the number of un-removed components for ``heavisine''.}
\label{fig:wl-box-plot-heavisine}
\end{figure}

\begin{figure}[p]
\begin{center}
\small
\includegraphics[width=90mm]{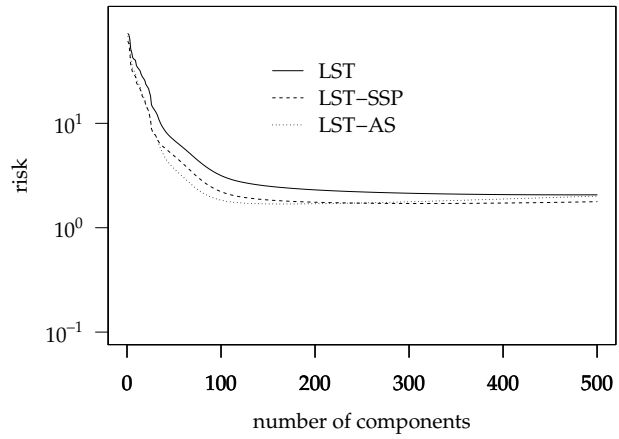} 

(a) Averaged risk curve of LST, LST-SSP and LST-AS.

\vspace{5mm}

\includegraphics[width=90mm]{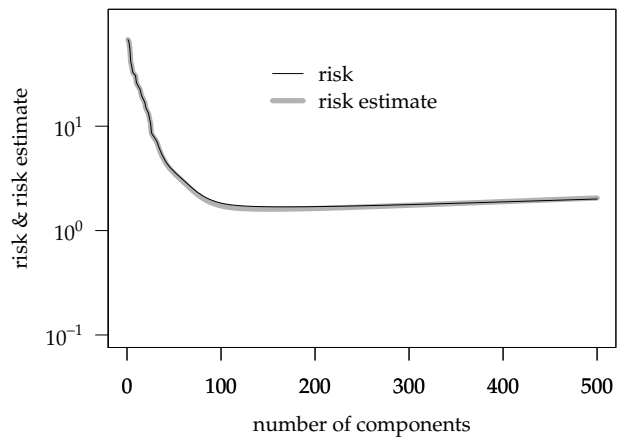} 

(b) Averaged risk curve and risk estimate of LST-AS.

\end{center}
\caption{Risk curve and risk estimate for ``blocks''.}
\label{fig:wl-risk-curve-blocks}
\end{figure}

\begin{figure}[p]
\begin{center}
\small
\includegraphics[width=90mm]{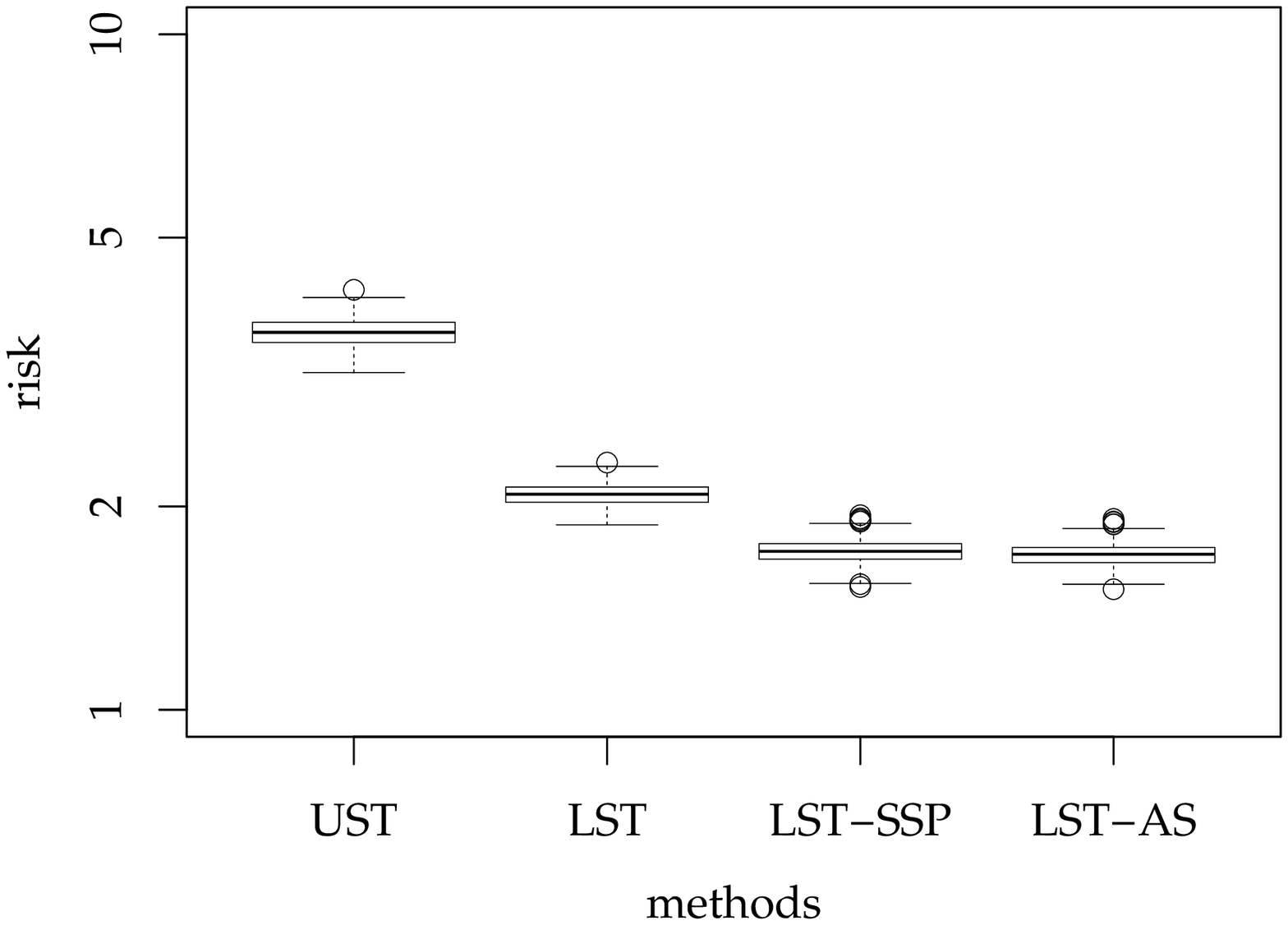} 

(a) Risk value at the selected number of components.

\vspace{5mm}

\includegraphics[width=90mm]{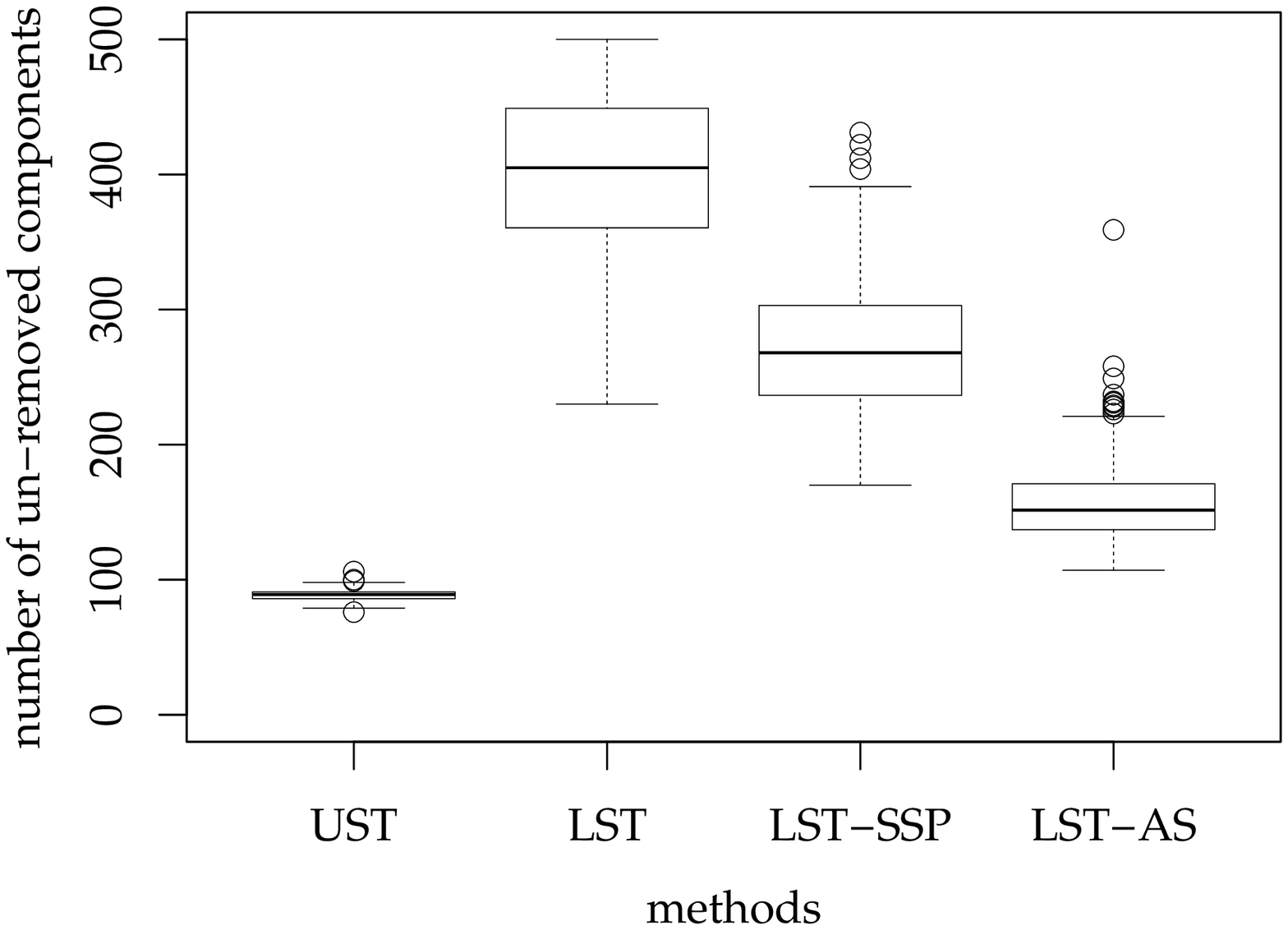} 

(b) The number of un-removed components.

\end{center}
\caption{Risk and the number of un-removed components for ``blocks''.}
\label{fig:wl-box-plot-blocks}
\end{figure}

\section{Conclusions and future works}

Soft-thresholding is a key modeling tool in statistical signal
processing such as wavelet denoising.  It has a parameter that
simultaneously controls threshold level and amount of shrinkage. This
parametrization is possible to suffer from an excess shrinkage for un-removed
valid components at a sparse representation; i.e. there is a dilemma
between prediction accuracy and sparsity.  In this paper, to overcome
this problem, we introduced a component-wise and data-dependent scaling method
for soft-thresholding estimators in a context of non-parametric
orthogonal regression including discrete wavelet transform. We refer
this method as an adaptive scaling method. Here, we employed a
LARS-based soft-thresholding method; i.e. a soft-thresholding method
that is implemented by LARS under an orthogonality condition. In
LARS-based soft-thresholding, a parameter value is selected by a
data-dependent manner by which a model selection problem reduces to the
determination of the number of un-removed components. We firstly derived
a risk given by LAR-based soft-thresholding estimate with our adaptive
scaling.  For determining an optimal number of un-removed components, we
then gave a model selection criterion as an unbiased estimate of the
risk. We also analyzed some properties of the risk curve and found that
the model selection criterion is possible to select a model with low
risk and high sparsity compared to a naive soft-thresholding. This was
verified by a simple numerical experiment and an application to wavelet
denoising.  As a future work, we need more application results. In doing
this, estimate of noise variance should be established in general
applications while MAD was found to be a good choice for a wavelet
denoising application.  Although
we gave scaling values in a top down manner in this paper, we may need
to test the other forms of adaptive scaling values; e.g. scaling values
which are estimates of optimal values in some senses. Moreover,
development of adaptive scaling for non-orthogonal case may be expected
for more general applications.

\appendix

\section{Lemmas}

We here give some lemmas that is used for proving the main theorems.

Let $X_1,\ldots,X_n$ be random variables.  We define the $m$th largest
value among $X_1,\ldots,X_n$ by $X_{(m)}=X_{(m)}(n)$.  

\begin{lemma}
\label{lemma:max-chi2-E}
Let $X_1,\ldots,X_n$ be i.i.d. random variables from $\chi^2_1$. 
We define $t_n=2\log n-\log\log n-\log\pi$. 
Then, at each fixed $k=1,2,\ldots$, 
\begin{align}
\label{eq:lemma-max-chi2-E}
\lim_{n\to\infty}\E\left((X_{(1)}-t_n)/2\right)^k=(-1)^k\Gamma^{(k)}(1)
\end{align}
hold, where $\Gamma^{(k)}(1)$ is the $k$th derivative of the Gamma
function at $1$. 
(\ref{eq:lemma-max-chi2-E}) implies that
\begin{align}
\label{eq:lemma-max-chi2-E-2}
\lim_{n\to\infty}\E\left[X_{(1)}^k/t_n^k\right]=1.
\end{align}
\end{lemma}

\begin{proof}{}
By slightly modifying Example 3, pp.72-73 in
\citep{SIR1987}, we can show that $\P\left\{
(X_{(1)}-t_n)/2\le x\right\}$ converges to the double
exponential distribution. Then, (\ref{eq:lemma-max-chi2-E}) is a direct
conclusion of Proposition 2.1 (iii) in \citep{SIR1987}. 
\end{proof}

\begin{lemma}
\label{lemma:prob-bound-mth-largest-chi2} 
Let $X_1,\ldots,X_n$ be i.i.d. random variables from $\chi^2_1$. 
At each fixed $m$, 
\begin{align}
\label{eq:prob-bound-mth-largest-chi2-lower} 
\lim_{n\to\infty}\P\left[X_{(m)}\le 2(1-\delta)\log n)\right]&=0\\
\label{eq:prob-bound-mth-largest-chi2-upper} 
\lim_{n\to\infty}\P\left[X_{(m)}>2\log n)\right]&=0
\end{align}
hold, where $\delta$ is an arbitrary positive constant.
\end{lemma}

\begin{proof}{}
We denote the probability distribution function of $\chi^2_1$ by $F_1$.
The probability density function of $\chi^2_1$ is
given by $f_1(x)=x^{-1/2}e^{-x/2}/\sqrt{2\pi}$. We have $\frac{d
f_1(x)}{dx}=-(1/x+1)f_1(x)/2$. Thus, we have
$\left(\int_x^{\infty}f_1(t)dt\right)/(2f_1(x))\to 1$ as $x\to\infty$ by
applying $\int_0^{\infty}f_1(t)dt=1$ and L'Hospital's rule. Therefore,
 for a $\chi_1^2$ random variable $X$, 
\begin{equation}
\label{eq:lemma-chi2-P-bound-1}
\P[X>x]\sim 2f_1(x) 
\end{equation}
holds for a sufficiently large $x$. 

By (\ref{eq:lemma-chi2-P-bound-1}), we obtain
\begin{align}
\P\left[X_{(1)}>2\log n\right]
&\le\sum_{i=1}^n\P\left[X_i>2\log n\right]\notag\\
&\sim 2nf_1\left(2\log n\right)\notag\\
&=\frac{1}{\pi}\frac{1}{\sqrt{\log n}}\to 0~(n\to\infty).
\end{align}
Since $X_{(1)}\ge X_{(m)}$ for any $m$, we have
(\ref{eq:prob-bound-mth-largest-chi2-upper}).
On the other hand, by (\ref{eq:lemma-chi2-P-bound-1}), we have
\begin{align}
n(1-F_1(2(1-\delta)\log n)))
&\sim 2nf_1(2(1-\delta)\log n)\notag\\
&\sim \frac{1}{\pi(1+\delta)}\frac{n^{\delta}}{\sqrt{\log n}}
\end{align}
for a sufficiently large $n$. Since this goes to $\infty$, we obtain
(\ref{eq:prob-bound-mth-largest-chi2-lower}) by Theorem 2.2.1 in \cite{LLR1983}.
\end{proof}

\begin{lemma}
\label{lemma:cj>maxci}
For any $j\in K^*$ and any $\rho>0$, 
\begin{equation}
\label{eq:cj>maxci}
\P\left[\tc_j^2\le\max_{i\in\oK^*}\tc_i^2\right]\le 2\pi^{-1/2}\rho^{-1/2}n^{-\rho} 
\end{equation}
holds for a sufficiently large $n$.
\end{lemma}

\begin{proof}{}
We define $\tau_{n,\rho}=2(\rho+1)\log n$. We obtain
\begin{align}
\label{eq:cj>maxci-0}
\P\left[\tc_j^2>\max_{i\in \oK^*}\tc_i^2\right]
&\ge\P\left[\left[\tc_j^2>\tau_{n,\rho}\right]
\bigcap\left[\tau_{n,\rho}>\max_{i\in \oK^*}\tc_i^2\right]\right]\notag\\
&=1-\P\left[\left[\tc_j^2\le\tau_{n,\rho}\right]
\bigcup\left[\tau_{n,\rho}\le\max_{i\in \oK^*}\tc_i^2\right]\right]\notag\\
&\ge 1-\P\left[\tc_j^2\le\tau_{n,\rho}\right]-
\P\left[\max_{i\in \oK^*}\tc_i^2\ge \tau_{n,\rho}\right].
\end{align}

By the definition of $\oc_j$, we have
\begin{align}
\label{eq:cj>maxci-1}
\P\left[\tc_j^2\le\tau_{n,\rho}\right]&=\P\left[|\tc_j|\le\sqrt{\tau_{n,\rho}}\right]\notag\\
&=\P\left[|\sqrt{n}\beta_j/\sigma+\oc_j|\le\sqrt{\tau_{n,\rho}}\right]\notag\\
&\le\P\left[\sqrt{n}|\beta_j|/\sigma-|\oc_j|\le \sqrt{\tau_{n,\rho}}\right]\notag\\
&=\P\left[|\oc_j|\ge\sqrt{n}|\beta_j|/\sigma-\sqrt{\tau_{n,\rho}}\right]\notag\\
&\le\P\left[|\oc_j|\ge\sqrt{\tau_{n,\rho}}\right]\notag\\
&=\P\left[\oc_j^2\ge\tau_{n,\rho}\right]
\end{align}
for a sufficiently large $n$. Note that this evaluation is not tight but
is enough in this paper. Since $\oc_j^2\sim\chi^2_1$ by the
definition of $\oc_j$, by (\ref{eq:lemma-chi2-P-bound-1}) and
(\ref{eq:cj>maxci-1}), we have
\begin{equation}
\label{eq:cj>maxci-2}
\P\left[\tc_j^2\le\tau_{n,\rho}\right]\le
\pi^{-1/2}\rho^{-1/2}n^{-\rho}
\end{equation}
for a sufficiently large $n$.

On the other hand, $\tc_i^2\sim\chi^2_1$ holds for $i\in \oK^*$ since
$\beta_i=0$ holds for $i\in \oK^*$. 
By (\ref{eq:lemma-chi2-P-bound-1}), we thus have
\begin{align}
\label{eq:cj>maxci-3}
\P[\max_{i\in \oK^*}\tc_i^2\ge\tau_{n,\rho}]
&\le\sum_{j\in\oK^*}\P[\tc_i^2\ge\tau_{n,\rho}]\notag\\
&\sim (n-k^*)\pi^{-1/2}(\rho+1)^{-1/2}n^{-(\rho+1)}\notag\\
&\le\pi^{-1/2}\rho^{-1/2}n^{-\rho}
\end{align}
for a sufficiently large $n$.
By (\ref{eq:cj>maxci-0}), (\ref{eq:cj>maxci-2}) and
 (\ref{eq:cj>maxci-3}), we obtain (\ref{eq:cj>maxci}) as desired.
\end{proof}

\begin{lemma}
\label{lemma:P-oEn*-bound}
\begin{equation}
\label{eq:P-oEn*-bound}
\P[\oE_n^*]\le k^*\pi^{-1/2}\rho^{-1/2}n^{-\rho}
\end{equation}
holds for any $\rho>0$ and a sufficiently large $n$.
\end{lemma}

\begin{proof}{}
If $E_n^*$ does not occur then there exist $l\in\{1,\ldots,k^*\}$ such
that $p_l\notin K^*$. This implies that there exist $j\in K^*$ and
$i\in\oK^*$ that satisfy $\tc_i^2\ge\tc_j^2$. Therefore, we have
$\oE_n^*\subseteq \bigcup_{j\in K^*}F_j$.  By Lemma
\ref{lemma:cj>maxci}, we then obtain (\ref{eq:P-oEn*-bound}).
\end{proof}

\begin{lemma}
\label{lemma:E-cp1^2-IoE*}
\begin{equation}
\label{eq:E-cp1^2-IoE*}
\lim_{n\to\infty}\E\tc_{p_1}^{2m}I_{\oE_n^*}=0
\end{equation}
holds for a fixed $m=1,2,\cdots$.
\end{lemma}

\begin{proof}{}
We define $\oc=\max_{1\le i\le n}|\oc_i|$ and $\obeta=\max_{j\in
 K^*}|\beta_j|$. We also define an event
 $F=\left\{\oc>(\obeta/\sigma)\sqrt{n}\right\}$.
By the Cauchy-Schwarz inequality, we have
\begin{align}
\label{eq:E-cp1^2-IoE*-2}
&\E[\tc_{p_1}^{2m}I_{\oE_n^*}]\notag\\
&\le\E[(\oc_{p_1}+(\beta_{p_1}/\sigma)\sqrt{n})^{2m}I_{\oE_n^*}]\notag\\
&\le\E[(\oc+(\obeta/\sigma)\sqrt{n})^{2m}I_{\oE_n^*}]\notag\\
&\le\E[(\oc+(\obeta/\sigma)\sqrt{n})^{2m}I_FI_{\oE_n^*}]
+\E[(\oc+(\obeta/\sigma)\sqrt{n})^{2m}I_{\oF}I_{\oE_n^*}]\notag\\
&\le 2^{2m}\E[\oc^{2m}I_FI_{\oE_n^*}]
+2^{2m}(\obeta/\sigma)^{2m}n^m\E[I_{\oF}I_{\oE_n^*}]\notag\\
&\le 2^{2m}\E[(\oc^2)^{m}I_{\oE_n^*}]
+2^{2m}(\obeta/\sigma)^{2m}n^m\E[I_{\oE_n^*}]\notag\\
&\le 2^{2m}\sqrt{\E[(\oc^2)^{2m}]}\sqrt{\P[\oE_n^*]}
+2^{2m}(\obeta/\sigma)^{2m}n^m\P[\oE_n^*].
\end{align}
By Lemma \ref{lemma:P-oEn*-bound} with $\rho>m$, the second term of
(\ref{eq:E-cp1^2-IoE*-2}) goes to zero as $n\to\infty$. Since $\oc^2$
is the largest value among i.i.d. $\chi^2_1$ sequence with size $n$, 
the first term of (\ref{eq:E-cp1^2-IoE*-2}) goes to zero as $n\to\infty$
 by Lemma \ref{lemma:max-chi2-E} and Lemma \ref{lemma:P-oEn*-bound} with
 the above choice of $\rho$.
\end{proof}

\begin{lemma}
\label{lemma:E-ttheta_k^2-bound} 
If $k\ge k^*$ then 
\begin{equation}
\label{eq:ttheta_k^2-upper-bound}
\E[\ttheta_k^{2m}]\le (2\log n)^m
\end{equation}
holds for a fixed $m=1,2,\cdots$ and sufficiently large $n$.
\end{lemma}

\begin{proof}{}
We can write
\begin{equation}
\label{eq:ttheta_k^2-upper-bound-1}
\E\left[\ttheta_k^{2m}\right]=\E\left[\ttheta_k^{2m}I_{\oE_n^*}\right]
+\E\left[\ttheta_k^{2m}I_{E_n^*}\right].
\end{equation}
By Lemma \ref{lemma:E-cp1^2-IoE*} and the definition of $\ttheta_k$,
\begin{equation}
\E\left[\ttheta_k^2I_{\oE_n^*}\right]\le\E\left[\tc_{p_1}^{2m}I_{\oE_n^*}\right]\to
 0~(n\to\infty).
\end{equation} 
We define $\tc^2=\max_{i\in\oK^*}\tc_i^2$.
If $E_n^*$ occurs then
 $\ttheta_k^2\le\tc^2$ and $\tc^2$
 is the largest value among i.i.d. $\chi_1^2$ random sequence with
 length $(n-k^*)$.  Therefore, by Lemma \ref{lemma:max-chi2-E},
\begin{equation}
\frac{\E\left[\ttheta_k^{2m}I_{E_n^*}\right]}{(2\log n)^m}
\le\frac{\E\left[\tc^{2m}\right]}{(2\log n)^m}\to 1~(n\to\infty).
\end{equation}
\end{proof}

\section{Proof of Theorems}

We give the proofs of the main theorems below.

\begin{proof}[Proof of Theorem \ref{theorem:Rnk-evalpha}]
For an $\evalpha$, the risk is reformulated as
\begin{align}
\label{eq:risk}
\RAS(n,k)&=\E\|\eAv\ebv_k-\vbeta\|^2\notag\\
&=\E\|\eAv\ebv_k-\ecv\|^2+\E\|\ecv-\h\|^2
+2\E(\eAv\ebv_k-\ecv)'(\ecv-\vbeta)\notag\\
&=\E\|\eAv\ebv-\ecv\|^2+\sigma^2+2\E(\eAv\ebv_k-\ecv)'(\ecv-\vbeta)\notag\\
&=\frac{1}{n}\E\|\evmu_{k,\alpha}-\y\|^2+\sigma^2+2\E(\eAv\ebv_k-\ecv)'(\ecv-\vbeta)\notag\\
&=\frac{1}{n}\E\|\evmu_{k,\alpha}-\y\|^2+\sigma^2+\frac{2\sigma^2}{n}
\E(\eAv\tbv_k-\tcv)'(\tcv-\c),
\end{align}
where we used (\ref{eq:dist-of-ecv}) at the third line and the
orthogonality condition at the last line. The last term is often called
the degree of freedom; see e.g. \citep{LARS}. 

Let $\tcv_{-j}$ be an $n-1$-dimensional vector that is constructed by
removing $\tc_j$ from $\tcv$. 
We define $d_j(\tcv)=(\ealpha_j\tb_{k,j}-\tc_j)$.
Although $d_j$ is a function
of $\tcv$, we regard this as a function $\tc_j$ under a fixed
$\tcv_{-j}$ and denote it by $d_j(\tc_j|\tcv_{-j})$. 
Let $\theta_j$ be the $k$th largest value in $\{|\tc_i|:i=1,\ldots,n,i\neq j\}$.
By (\ref{eq:component-wise-scaling}), we have
\begin{equation}
d_j(\tc_j|\tcv_{-j})=\ealpha_j\tb_{k,j}-\tc_j=
\begin{cases}
-\theta_j^2/\tc_j & |\tc_j|>\theta_j\\
-\tc_j & |\tc_j|\le\theta_j\\		 
\end{cases}.
\end{equation}
Note here that $d_j(\tc_j|\tcv_{-j})$ is well-defined even when $\tc_j=0$
 under the definition of $\ealpha_j$ in (\ref{eq:component-wise-scaling}).
This is Lipschitz continuous as a function of $\tc_j$ when $\tcv_{-j}$
is fixed. It is thus absolutely continuous. 
On the other hand, we denote expectation with respect to $\tcv$ by $\E_{\tcv}$. We have
$\tcv=\G'\y/(\sigma\sqrt{n})$ and
$|\det(\G'\y/(\sigma\sqrt{n}))^{-1})|=\sigma^n$, where $\det$
denotes the determinant of a matrix.
Therefore, $\E$ is always replaced with $\E_{\tcv}$ by change of variables.
We also denote a conditional expectation with respect to $\tc_j$ given
$\tcv_{-j}$ by $\E_{\tc_j|\tcv_{-j}}$.
We define $I_{[-\theta_j,\theta_j]}(\tc_j|\tcv_{-j})$ by 
$I_{[-\theta_j,\theta_j]}(\tc_j|\tcv_{-j})=1$ when
$\tc_j\in[-\theta_j,\theta_j]$ and $0$ otherwise under a fixed
$\tcv_{-j}$. 
Then, by applying this change of variables and Stein's
lemma\citep{CS1981} under the above absolutely continuity, we obtain
\begin{align}
\label{eq:A-tbv-tcv-tcv-c}
&\E(\eAv\tbv_k-\tcv)'(\tcv-\c)\notag\\
&=\sum_{j=1}^n\E\left[d_j(\tcv)(\tc_j-c_j)\right]\notag\\
&=\sum_{j=1}^n\E_{\tcv}\left[d_j(\tcv)(\tc_j-c_j)\right]\notag\\
&=\sum_{j=1}^n\E_{\tcv_{-j}}\E_{\tc_j|\tcv_{-j}}
\left[d_j(\tc_j|\tcv_{-j})(\tc_j-c_j)\right]\notag\\
&=\sum_{j=1}^n\E_{\tcv_{-j}}\E_{\tc_j|\tcv_{-j}}
\left[\frac{\partial d_j(\tc_j|\tcv_{-j})}{\partial
 \tc_j}\right]\notag\\
&=\E_{\tcv}\left[\sum_{j=1}^n
(\theta_j^2/\tc_j^2)(1-I_{[-\theta_j,\theta_j]}(\tc_j|\tcv_{-j}))\right]
-\E_{\tcv}\left[
\sum_{j=1}^nI_{[-\theta_j,\theta_j]}(\tc_j|\tcv_{-j})\right]\notag\\
&=\E\left[\sum_{j\in\eK_k}(\ealpha_j-1)^2\right]-(n-k),
\end{align}
where the last line is obtained by the definition of $\theta_j$ and
 $\ealpha_j$.
(\ref{eq:risk}) and 
(\ref{eq:A-tbv-tcv-tcv-c}) yield (\ref{eq:theorem-Rnk-evalpha}).
\end{proof}

\begin{proof}[Proof of Theorem \ref{theorem:ealpha_j-convergence}]
We show that 
\begin{align}
\label{eq:ealpha_j-bound-in*-1}
\P\left[\ealpha_j>1+\epsilon_{j,n}\right]\to 0~(n\to\infty)
\end{align}
for $j\in K^*$. 
We define $E_0=\{\tc_j=0\}$ for which $\P[E_0]=0$.
By the definition of $\ealpha_j$ in
 (\ref{eq:component-wise-scaling}), we then have
\begin{align}
\label{eq:ealpha_j-bound-in*-2}
&\P\left[\ealpha_j>1+\epsilon_{j,n}\right]\notag\\
&=\P\left[\left\{\ealpha_j>1+\epsilon_{j,n}\right\}\bigcap \oE_0\right]\notag\\
&=\P\left[\left\{\frac{\ttheta_k}{|\tc_j|}>\epsilon_{j,n}\right\}\bigcap \oE_0\right]\notag\\
&\le\P\left[\frac{\ttheta_k}{\sqrt{2\log n}}
-\frac{|\tc_j|}{(|\beta_j|+\delta)\sqrt{n}}>0\right]\notag\\
&\le
\P\left[\ttheta_k>\sqrt{2\log n}\right]
+\P\left[|\tc_j|<(|\beta_j|+\delta)\sqrt{n}\right].
\end{align}
For the first term of (\ref{eq:ealpha_j-bound-in*-2}), we have
\begin{align}
\label{eq:ealpha_j-bound-in*-3}
&\P\left[\ttheta_k>\sqrt{2\log n}\right]\notag\\
&=\P\left[\ttheta_k^2>2\log n\right]\notag\\
&=\P\left[\ttheta_k^2>2\log n|E_n^*\right]\P\left[E_n^*\right]+
\P\left[\ttheta_k^2>2\log n|\oE_n^*\right]\P\left[\oE_n^*\right]\notag\\
&\le\P\left[\ttheta_k^2>2\log n|E_n^*\right]+
\P\left[\oE_n^*\right].
\end{align}
The second term of (\ref{eq:ealpha_j-bound-in*-3}) goes
to zero as $n\to\infty$ by Lemma \ref{lemma:P-oEn*-bound}. If $E_n^*$
occurs then $\ttheta_k^2$ is the $(k+1-k^*)$th largest value among
i.i.d. $\chi_1^2$ random sequence with size $n-k^*$.  Therefore, by
(\ref{eq:prob-bound-mth-largest-chi2-upper}) in Lemma
\ref{lemma:prob-bound-mth-largest-chi2}, the first term of
(\ref{eq:ealpha_j-bound-in*-3}) goes to zero as $n\to\infty$. Thus, the
first term of (\ref{eq:ealpha_j-bound-in*-2}) goes to zero as
$n\to\infty$.  Recall that $\tc_j=\sqrt{n}\beta_j+\oc_j$ for $j\in K^*$,
where $\oc_j\sim N(0,1)$. Then, for the second term of
(\ref{eq:ealpha_j-bound-in*-2}), we obtain
\begin{align}
\label{eq:ealpha_j-bound-in*-4} 
\P\left[|\tc_j|<(|\beta_j|+\delta)\sqrt{n}\right]
&=\P\left[|\sqrt{n}\beta_j+\oc_j|<(|\beta_j|+\delta)\sqrt{n}\right]\notag\\
&\le\P\left[\sqrt{n}|\beta_j|-|\oc_j|<(|\beta_j|+\delta)\sqrt{n}\right]\notag\\
&=\P\left[|\oc_j|>\delta\sqrt{n}\right]\to 0~(n\to\infty).
\end{align}
Since $\ealpha_j\ge 1$ holds, we obtain
(\ref{eq:ealpha_j-lower-bound-K*-0}) as desired.

On the other hand, we consider (\ref{eq:ealpha_j-lower-bound-oK*}). 
For any $\delta_n$, we have
\begin{align}
\label{eq:ealpha_j-bound-notin*-1}
\P\left[\ealpha_j\le 2-\epsilon\right]
&=\P\left[\{\ealpha_j\le 2-\epsilon\}\bigcap\oE_0\right]\notag\\
&\le\P\left[\ttheta_k\le (1-\epsilon)|\tc_j|\right]\notag\\
&\le\P\left[\ttheta_k\le\delta_n\right]
+\P\left[(1-\epsilon)|\tc_j|>\delta_n\right].
\end{align}
For the first term of (\ref{eq:ealpha_j-bound-notin*-1}), we have
\begin{align}
\label{eq:ealpha_j-bound-notin*-2}
\P\left[\ttheta_k\le\delta_n\right]
&=\P\left[\ttheta_k^2\le\delta_n^2\right]\notag\\
&=\P\left[\ttheta_k^2\le\delta_n^2|E_n^*\right]\P[E_n^*]+
\P\left[\ttheta_k\le\delta_n^2|\oE_n^*\right]\P[\oE_n^*]\notag\\
&\le\P\left[\ttheta_k^2\le\delta_n^2|E_n^*\right]+\P[\oE_n^*].
\end{align}
By Lemma \ref{lemma:P-oEn*-bound}, the second term of
(\ref{eq:ealpha_j-bound-notin*-2}) goes to zero as $n\to\infty$.  We set
$\delta_n=\sqrt{2(1-\epsilon)\log n}$. If $E_n^*$ occurs then
$\ttheta_k^2$ is the $(k+1-k^*)$th largest value among i.i.d. $\chi_1^2$
random sequence with size $n-k^*$.  Therefore, by
(\ref{eq:prob-bound-mth-largest-chi2-lower}) in Lemma
\ref{lemma:prob-bound-mth-largest-chi2} and the choice of $\delta_n$, the
first term of (\ref{eq:ealpha_j-bound-notin*-2}) goes to zero as
$n\to\infty$.  
We define $\tc^2=\max_{i\in\oK^*}\tc_i^2$.
For the second term of (\ref{eq:ealpha_j-bound-notin*-1}), we have
\begin{align}
\label{eq:ealpha_j-bound-notin*-3}
\P\left[(1-\epsilon)|\tc_j|>\delta_n\right]\le\P\left[\tc^2>2\log n\right]
\end{align}
since $j\in K^*$. Here, $\tc^2$ is the largest value among
i.i.d. $\chi_1^2$ random sequence with size $n-k^*$ by the definitions
of $\tc_i$ and $K^*$. Hence, by 
(\ref{eq:prob-bound-mth-largest-chi2-upper}) in Lemma
\ref{lemma:prob-bound-mth-largest-chi2}, 
 (\ref{eq:ealpha_j-bound-notin*-3}) goes to zero as $n\to\infty$.
\end{proof}

\begin{proof}[Proof of Theorem \ref{theorem:R(1)-R(evalpha)}]
By (\ref{eq:risk-for-LST}) and (\ref{eq:theorem-Rnk-evalpha}), we have
\begin{align}
\label{eq:theorem-R(1)-R(evalpha)-1}
&R_{n,k^*}(\1_n)-\RAS(n,k^*)\notag\\
&=\frac{\sigma^2}{n}\sum_{j=1}^{k^*}\left(\E\left[(\ealpha_{p_j}-1)^2\tc_{p_j}^2\right]
-\E\left[(\ealpha_{p_j}-1)^2\ttheta_{k^*}^2\right]
-2\E\left[(\ealpha_{p_j}-1)^2\right]\right)\notag\\
&=\frac{\sigma^22\log n}{n}
\sum_{j=1}^{k^*}\left(\frac{\E\left[\ttheta_{k^*}^2\right]}{2\log n}
-\frac{\E\left[(\ealpha_{p_j}-1)^2\ttheta_{k^*}^2\right]}{2\log n}
-2\frac{\E\left[(\ealpha_{p_j}-1)^2\right]}{2\log n}\right)
\end{align}
through a simple calculation.
We evaluate the three terms in the sum of (\ref{eq:theorem-R(1)-R(evalpha)-1}).
We first have 
\begin{equation}
\lim_{n\to\infty}\E[\ttheta_{k^*}^2]/(2\log n)=1
\end{equation}
by Lemma \ref{lemma:E-ttheta_k^2-bound}.
Hence, the proof is completed by showing that the second and third terms of 
(\ref{eq:theorem-R(1)-R(evalpha)-1}) goes to zero as $n\to\infty$.
We define $\epsilon_n=\max_{j\in K^*}\epsilon_{j,n}$ and
$G_j=\{(\ealpha_{p_j}-1)^2>\epsilon_n^2\}$, where $\epsilon_{j,n}$ is
 defined in (\ref{eq:def-epsilon_n}). 
We have $(\ealpha_{p_j}-1)^2\le 1$ for $j\le k^*$ by the definition of
 $\ttheta_{k^*}=|\tc_{p_{k^*+1}}|$. And, if $E_n^*$ occurs then
 $\bigcup_{l\in K^*}\{p_j=l\}$ for any $j\in\{1,\ldots,k^*\}$.
We then obtain
\begin{align}
\label{eq:theorem-R(1)-R(evalpha)-3}
\E\left[(\ealpha_{p_j}-1)^2\right]
&=\E\left[(\ealpha_{p_j}-1)^2I_{G_j\bigcap E_n^*}\right]
+\E\left[(\ealpha_{p_j}-1)^2I_{\oG_j\bigcup\oE_n^*}\right]\notag\\
&\le\E\left[I_{G_j\bigcap\oE_n^*}\right]+\epsilon_n^2+\P[\oE_n^*]\notag\\
&\le\sum_{l\in K^*}\P\left[(\ealpha_l-1)^2>\epsilon_n^2\right]+\epsilon_n^2+\P[\oE_n^*].
\end{align}
(\ref{eq:theorem-R(1)-R(evalpha)-3}) goes to zero as $n\to\infty$ by 
(\ref{eq:ealpha_j-lower-bound-K*-0}) in Theorem
 \ref{theorem:ealpha_j-convergence}, the definition of $\epsilon_n$ and Lemma 
\ref{lemma:P-oEn*-bound}. We also have
\begin{align}
\label{eq:theorem-R(1)-R(evalpha)-4}
&\E\left[(\ealpha_{p_j}-1)^2\ttheta_{k^*}^2\right]\notag\\
&=\E\left[(\ealpha_{p_j}-1)^2\ttheta_{k^*}^2I_{G_j\bigcap E_n^*}\right]
+\E\left[(\ealpha_{p_j}-1)^2\ttheta_{k^*}^2I_{\oG_j\bigcup\oE_n^*}\right]\notag\\
&\le\E\left[\ttheta_{k^*}^2I_{G_j}I_{E_n^*}\right]
+\epsilon_n^2\E[\ttheta_{k^*}^2]+\E[\ttheta_{k^*}^2I_{\oE_n^*}].
\end{align}
For the first term of (\ref{eq:theorem-R(1)-R(evalpha)-4}), by the
Cauchy-Schwarz inequality, we have
\begin{align}
\label{eq:theorem-R(1)-R(evalpha)-5}
\frac{\E\left[\ttheta_{k^*}^2I_{G_j}I_{E_n^*}\right]}{2\log n}
&\le\frac{\sqrt{\E\left[\ttheta_{k^*}^4\right]}}{2\log n}
\sqrt{\E\left[I_{G_j}I_{E_n^*}\right]}\notag\\
&\le\frac{\sqrt{\E\left[\ttheta_{k^*}^4\right]}}{2\log n}
\sqrt{\sum_{l\in K^*}\P\left[(\ealpha_l-1)^2>\epsilon_n^2\right]}.
\end{align}
(\ref{eq:theorem-R(1)-R(evalpha)-5}) goes to zero as $n\to\infty$ by
Lemma \ref{lemma:E-ttheta_k^2-bound} and
(\ref{eq:ealpha_j-lower-bound-K*-0}) in Theorem
\ref{theorem:ealpha_j-convergence}.  The second term of
(\ref{eq:theorem-R(1)-R(evalpha)-4}) goes to zero as $n\to\infty$ by
Lemma \ref{lemma:E-ttheta_k^2-bound} and the definition of $\epsilon_n$.
By the Cauchy-Schwarz
inequality, the third term of (\ref{eq:theorem-R(1)-R(evalpha)-4}) is bounded above
by $\sqrt{\E[\ttheta_{k^*}^4]}\sqrt{\P[\oE_n^*]}$. This goes to zero as $n\to\infty$ by Lemma
\ref{lemma:E-ttheta_k^2-bound} and Lemma \ref{lemma:P-oEn*-bound}.
We thus obtain (\ref{eq:theorem-R(1)-R(evalpha)}) as desired.
\end{proof}

\end{document}